\documentclass[letterpaper, 10 pt, conference]{ieeeconf} 

\IEEEoverridecommandlockouts                              

\usepackage{lipsum}
\usepackage{graphicx}   
\usepackage{subcaption}    
\usepackage{amsmath}
\usepackage{amssymb} 
\usepackage{float}
\usepackage{xcolor}
\usepackage{tikz}
\usetikzlibrary{automata}
\usetikzlibrary{shapes,arrows}
\newtheorem{defn}{Definition}

\newtheorem{thm}{Theorem}
\newtheorem{rmk}{Remark}
\newtheorem{myeg}{Example}
\newtheorem{lem}{Lemma}

\newcommand{\PP}{\mathbb{P}}
\newcommand{\EE}{\mathbb{E}}
\newcommand{\RR}{\mathbb{R}}
\newcommand{\NN}{\mathbb{N}}

\newcommand{\begit}{\begin{itemize}}
	\newcommand{\eit}{\end{itemize}}
\newcommand{\bseq}{\begin{subequations}}
	\newcommand{\eseq}{\end{subequations}}
\newcommand{\bpat}{\begin{pmatrix}}
	\newcommand{\epat}{\end{pmatrix}}
\newcommand{\bmat}{\begin{bmatrix}}
	\newcommand{\emat}{\end{bmatrix}}
\newcommand{\beq}{\begin{equation}}
\newcommand{\eeq}{\end{equation}}
\newcommand{\bc}{\begin{cases}}
\newcommand{\ec}{\end{cases}}
\newcommand{\beqs}{\begin{equation*}}
\newcommand{\eeqs}{\end{equation*}}
\usepackage{hyperref}
\title{\LARGE \bf
Sufficient Lyapunov conditions for exponential mean square stability of discrete-time systems with markovian delays
}
\author{Anastasia Impicciatore$^{1}$,~Maria Teresa Grifa$^{1}$,~Pierdomenico Pepe$^{1,2}$,~Alessandro D'Innocenzo$^{1,3}$
\thanks{$^1$Department of Information Engineering, Computer Science, and Mathematics, University of L'Aquila, Via Vetoio, Loc.~Coppito, 67100 L'Aquila, Italy.}%
\thanks{$^2$Center of Excellence DEWS, Via Vetoio, Loc. Coppito, 67100 L'Aquila, Italy.}%
\thanks{$^3$Center of Excellence EX-EMERGE, Via Vetoio, Loc. Coppito, 67100 L'Aquila, Italy.}
\thanks{E.email: anastasia.impicciatore@graduate.univaq.it,~mariateresa.grifa@graduate.univaq.it,~pierdomenico.pepe@univaq.it,~alessandro.dinnocenzo@univaq.it}%
}
\begin{document}
\maketitle
\thispagestyle{empty}
\pagestyle{empty}
\begin{abstract}
This paper introduces sufficient Lyapunov conditions guaranteeing exponential mean square stability of discrete-time systems with markovian delays.
We provide a transformation of the discrete-time system with markovian delays into a discrete-time Markov jump system. Then, we extend sufficient Lyapunov conditions existing for the global asymptotic stability of discrete-time systems with delays digraphs to the mean square stability of discrete-time systems with markovian delays.   
Finally, an example is provided to illustrate the efficiency and advantage of the proposed method.
\end{abstract}
%
\section{INTRODUCTION}
\noindent This paper aims to study the nonlinear discrete-time delay systems with delays constrained to vary on a Markov chain (see \cite{Bortolin2019} for the linear case).
The stability analysis for discrete-time delay system is studied in \cite{leite2008robust, fridman2014introduction, gielen2013necessary, chen2016robust, Pepe2018OnLyap, silva2018robust,
baker2010development,  liu2007razumikhin, Pepe2019converse, de2018iss, liu2009input,hetel2008equivalence}.
Time-delays often lead to complex behaviors in the dynamics of a system and may lead to the failure of stability.
Constraints on time-delays can be described by means of the delays digraphs notion (see \cite{Pepe2020,grifa2020stability}). 
In recent years, the graph theory approach has been satisfactorily used in the development of stability theory for discrete-time switching systems with constrained switching signals (see \cite{Pepe2019converse, athanasopoulos2014stability, kundu2016graph, cao2020} and the references therein). 
The  motivations of modeling the constraints through a digraph
and the impact of this choice in establishing the stability results are presented in \cite{Pepe2020}. Constraints provided by bounded delay variations are studied in \cite{ Bortolin2019,de2018iss, Zhang2016}. 
In \cite{Bortolin2019}, the regulation problem for discrete-time linear systems with bounded unknown random state delay is presented.
In \cite{Zhang2016}, the problem of disturbance rejection control for markovian jump linear systems is investigated. 
The modelling framework for systems subject to markovian switching is given by discrete-time markovian switching systems, also known as {\it Markov jump systems}. There is a wide literature investigating this kind of systems. Discrete-time markovian switching systems are particularly useful in the modeling of systems subject to abrupt changes, such as Wireless Control Networks (WCNs). 
Markovian switching systems are good approximations of the stochastic characterization of WCNs models in presence of packet losses and induced random delays.
In \cite{Alur2011compositional,%
DiGirolamo2019codesign,%
Lun2019stabilizability,%
Lun2020ontheimpact} the use of markovian switching systems handles the challenges in analysis and co-design of wireless networked control systems, and allows to verify instability of a system due to bursts of packet loss when Bernoulli-like channel models fail. 
Moreover, the Markov modelling of the Wireless channel (see \cite{DiGirolamo2019codesign,%
Lun2019stabilizability,Sadeghi2008finite}) allows performance improvement in stabilizing control synthesis, as it is shown in \cite{Lun2020ontheimpact}.
Motivated by the above discussions, we aim to study discrete-time systems with markovian delays, linking the methodologies available for Markov jump systems and discrete-time systems with constrained delays. 
The link between discrete-time systems with delays and switching delay-free systems is provided in \cite{hetel2008equivalence,Pepe2020%
}.
We provide an exponential mean square stability  analysis for the class of considered systems, i.e., a stability analysis concerning the behaviour of the second moment of the state.
 The mean square stability of discrete-time markovian switching systems has been extensively analysed in the linear case (see \cite{COSTA2005,Lun2019robust}),
only few works presented in the literature investigate this stability notion in the nonlinear framework (see \cite{Patrinos2014stochastic,%
Tejada2010onnonlinear}). On performing the analysis,
we write the discrete-time delay system as a switching 
discrete-time system where the delays are constrained to adhere to a Markov chain. Then, we transform the switching systems with markovian delays to a Markov jump system.
The Lyapunov conditions guaranteeing global asymptotic stability of discrete-time delay systems with delays switches digraphs already exist in literature (see \cite{Pepe2020}). Sufficient Lyapunov conditions guaranteeing exponential mean square stability of discrete-time Markov jump systems are introduced in \cite{Patrinos2014stochastic,%
Impicciatore2020sufficient}.
Our contribution consists in the extension of the Lyapunov conditions
in \cite{Pepe2020} for discrete-time systems with delays digraph,
to the study of exponential mean square stability of discrete-time systems with markovian delays. 
We provide a methodology which makes use of multiple Lyapunov functions (see \cite{Pepe2019converse,Impicciatore2020sufficient}) depending on the mode of the Markov chain, that governs the switching delay. 
The remainder of the paper is organized as follows. In Section II, discrete-time systems with markovian delay signals are introduced.
In Section III, we provide the main result of the paper consisting of sufficient Lyapunov conditions guaranteeing exponential mean square stability. In Section IV, we illustrate a meaningful example showing the effectiveness of our result. Conclusions are provided in Section \ref{sec:conclusion}. The proofs are reported in the Appendix.
\subsection{Notation and basic definitions}
\noindent The symbols {\small$\NN,$} {\small$\RR,$} and {\small$\RR^+$} denote the set of non-negative integer numbers, the set of real numbers, and the set of non-negative real numbers, respectively.
For a given finite set {\small $D,$} {\small $card(D)$} denotes its cardinality. The notation \mbox{\small $\lVert x \rVert $} is used to denote the Euclidean norm of a vector \mbox{\small $x \in \RR^{n}.$} 
For any positive real \mbox{\small $\Delta$} and any positive integer \mbox{\small $n,$} the symbol \mbox{\small $\mathcal{C}$} denotes the space of functions mapping
\mbox{\small $\{ -\Delta, -\Delta+1, \dots, 0 \}$} into \mbox{\small $
\mathbb{R}^{n}.$} 
For \mbox{\small $\phi \in \mathcal{C},$ $\lVert \phi \rVert_{\infty} = \max\{ \lVert\phi(-j)\rVert : j= 0, 1, \dots, \Delta \}.$} For any non-negative integer \mbox{\small $c,$} (or for \mbox{\small $c=+\infty$}), for any function \mbox{\small $x : \{ -\Delta, -\Delta+1,\dots, c \} \rightarrow \mathbb{R}^{n},$} for any integer \mbox{\small $k \in [0, c] \cap \mathbb N$, $x_{k}$} is the function in \mbox{\small $\mathcal{C}$} defined, for \mbox{\small $\tau \in\{ -\Delta, -\Delta+1, \dots, 0\},$} as \mbox{\small $x_{k}(\tau) = x(k+\tau).$} 
\noindent
The function \mbox{\small$\mathrm{sat}:\RR\to [-1,1]$} is defined, for \mbox{\small $x\in\RR,$} as \mbox{\small $\mathrm{sat}(x)=\min\{1,\max\{x,-1\}\}.$}
\noindent We consider the stochastic basis defined by the quadruple \mbox{\small $\left(\Omega,\mathcal{G}, \{\mathcal{G}_k\},\mathbb{P}\right),$} where \mbox{\small $\Omega$} is the sample space, \mbox{\small $\mathcal{G}$} is the corresponding $\sigma$-algebra of events, \mbox{\small $\{\mathcal{G}_k\}_{k\in\NN}$} is the filtration, \mbox{\small $\mathbb{P}$} is the probability measure.
Let \mbox{\small $\EE\left[\cdot\right]$} denote the expectation of a random variable with respect to \mbox{\small $\mathbb{P},$} and let  \mbox{\small $\EE\left[\; \cdot\;|\mathcal{G}_k\right]$} denote the conditional expectation of a random variable on the filtration \mbox{\small $\{\mathcal{G}_k\}_{k\in\NN}.$}
\noindent The acronyms {\it EMSS} and {\it GAS} stand for exponential mean square stability or exponentially mean square stable, and global asymptotic stability or globally asymptotically stable, respectively.
\section{DISCRETE - TIME SYSTEMS WITH MARKOVIAN DELAYS}
\noindent Let us consider the discrete-time delay system \mbox{\small$\mathcal D$} of the form (see \cite{Pepe2020})
\begin{align}\label{eq:plain time delay system}
	x(k+1) &= f(x(k), x(k-d_{1}(k)),\dots , x(k-d_{r}(k)) ),  \nonumber\\ 
	x(\theta) &= \xi_{0}(\theta), \; \theta \in \{ -\Delta, -\Delta + 1, \dots, 0\},
\end{align}
where: \mbox{\small $k\in \NN;$} \mbox{\small $\Delta$} is a known positive integer, the maximum involved time delay; \mbox{\small $x(j) \in \RR^{n},$} \mbox{\small $j\ge -\Delta;$} for \mbox{\small $1 \leq i \leq r,$} \mbox{\small $d_{i}(k) \in \{0,1,\dots, \Delta \}$} is a time-varying time delay, \mbox{\small $r$} is a known positive integer; the function \mbox{\small $f:\RR^{n(r+1)} \rightarrow  \RR^{n}$}
satisfies the equality \mbox{\small $f(0,0, \dots, 0) = 0;$}   \mbox{\small $\xi_{0}\in \mathcal{C}{.}$}
Let \mbox{\small $d(k) = [ d_{1}(k) \;  d_{2}(k) \; \dots \; d_{r}(k)]^{T}{,}$}  \mbox{\small $k \in \NN,$} denote the vector collecting all time delays at time \mbox{\small $k.$}
Let \mbox{\small $D\subset\{0, 1,\dots,\Delta\}^{r}$} be the set of allowed values for the time-delays vector \mbox{\small $d(k).$} 
That is, for any \mbox{\small $k \in \NN,$} \mbox{\small $d(k) \in D.$}
\noindent The system \eqref{eq:plain time delay system} can be rewritten by using the following equation, (see \cite{Pepe2018OnLyap} and the references therein):
\begin{align}\label{eq:Fofxk}
 x_{k+1} &= F(x_k, d(k)),\quad k\in\NN,\nonumber\\
 x_0 &= \xi_0,\quad \xi_0\in\mathcal{C},
\end{align}
\noindent \mbox{\small $x_k\in\mathcal{C},\,x_k(\theta)=x(k+\theta),\,\theta\in\{-\Delta,-\Delta+1,\ldots,0\},$} \mbox{\small $k\in\NN.$} The map \mbox{\small $F:\mathcal{C}\times D\to \mathcal{C}$} is defined, for \mbox{\small $\phi \in \mathcal{C},$} \mbox{\small$d=\bmat d_1 & d_2 &\ldots & d_r\emat^T\in D,$}
\begin{align}\label{eq:Fofphid}
    F(\phi, d ) (\theta) =
    \bc
    f(\phi(0),\phi(-d_1),\ldots, \phi(-d_r)),\,  \theta=0,\\
    \phi(\theta+1),\, \theta = -\Delta, -\Delta+1, \ldots, -1.
    \ec
\end{align}
\noindent Let us define the Markov chain (hereafter MC) as \mbox{\small $\eta:\NN\to \mathcal{S},$} with \mbox{\small $\mathcal{S}\triangleq \{1,2,\ldots,s\},$} \mbox{\small$s=card(D).$} The transition probability matrix (hereafter TPM) of the MC is defined as 
\bseq\label{eq:TPMdef}
\beq\label{eq:TPMdef_1}
P\triangleq\bmat p_{ij}\emat_{i,j\in\mathcal{S}},\quad p_{ij}\triangleq \PP(\eta(k+1)=j|\eta(k)=i),
\eeq
\noindent for all $i,j\in\mathcal{S}$, and
\beq\label{eq:TPMdef_2}
{\displaystyle \sum_{j\in\mathcal{S}}} p_{ij}=1,\,\forall i\in\mathcal{S},\quad 0\leq p_{ij}\leq 1,\quad\forall i,j\in\mathcal{S}.
\eeq
\eseq
\noindent Assume that the delay \mbox{\small $d(k+1),\,k\in\NN,$} depends only on the delay at the previous step \mbox{\small $d(k),$} \mbox{\small $k\in\NN,$} and assume that our prior knowledge on the transition from \mbox{\small $d(k)$} to \mbox{\small $d(k+1)$} is given by a transition probability.  \\
\noindent Let \mbox{\small $H:D \rightarrow \mathcal{S}$} be a bijective function defined for all \mbox{\small $\delta_i\in D,$} and for all \mbox{\small $i\in\mathcal{S},$} as
\beq\label{eq:bijective_H}
H(\delta_i)\triangleq i.
\eeq
\noindent The inverse function of \mbox{\small $H$} is \mbox{\small $H^{-1}: \mathcal{S} \rightarrow D,$} defined for all \mbox{\small $i\in\mathcal{S}$} and for all \mbox{\small $\delta_i\in D,$} as follows
\beq\label{eq:bijective_inverseH}
H^{-1}(i)\triangleq \delta_i.
\eeq
\noindent Consider \mbox{\small $p_{ij}$} defined in \eqref{eq:TPMdef_1}. By applying the definition of \mbox{\small $p_{ij}$} in \eqref{eq:TPMdef_1} and the definition of the functions \mbox{\small $H$} and \mbox{\small $H^{-1},$} the following equalities hold:
\begin{align}\label{eq:pdeltai_deltaj}
p_{ij}&=\PP\left( \eta(k+1)=j|\eta(k)=i \right)\nonumber\\
&=\PP\Big ( H(d(k+1))=H(\delta_j)|H(d(k))=H(\delta_i) \Big)\nonumber\\
&=\PP\left(d(k+1)=\delta_j|d(k)=\delta_i\right),
\end{align}
\noindent for all \mbox{\small $\delta_i,\delta_j\in D,$} for all \mbox{\small $i,j\in\mathcal{S}.$}\\
\noindent Consequently, the modes of the MC \mbox{\small $\{\eta(k)\}_{k\in\NN}$} with TPM defined by \eqref{eq:TPMdef} are associated with the delays in the set \mbox{\small $D,$} through the function \mbox{\small $H^{-1}.$}\\
\noindent Let \mbox{\small $\mathcal{E}(D)$} be the finite set of all pairs \mbox{\small $(\delta_i, \delta_j) \in D \times D,\,i,j\in\mathcal{S},$} such that, for any
\mbox{\small $k\in \NN,$} if \mbox{\small $d(k) = \delta_i,$} it is allowed \mbox{\small$d(k + 1) = \delta_j.$} We define the set \mbox{\small$\mathcal{E}(D)$} as follows,
\beq\label{eq:mathcalE_def}
\mathcal{E}(D)\triangleq\{(\delta_i,\delta_j)\in D\times D,\,\delta_i,\delta_j \in D,\, i,j\in\mathcal{S} \,|\, p_{ij}>0\}.
\eeq
\noindent In the following we define the Markov jump system that we consider throughout the paper (see \cite{COSTA2005,%
Lun2019robust} and the references therein).
\begin{defn}\label{def:Markov_jump_system}
\noindent Let {\small $\Sigma$} denote the Markov jump system defined on the stochastic basis \mbox{\small $\left(\Omega,\mathcal{G},\mathcal{G}_k,\mathbb{P}\right),$} as 
\begin{align}\label{eq:Sigma}
\Sigma\triangleq\left(\mathcal{D},P,H\right),
\end{align}
\noindent where \mbox{\small $\mathcal{D}$} is the system described by \eqref{eq:plain time delay system} and rewritten in the form \eqref{eq:Fofxk}, \mbox{\small $P$} is a known TPM defined by \eqref{eq:TPMdef} modeling the stochastic switching of the delays, and \mbox{\small $H$} is the bijective function defined by \eqref{eq:bijective_H}.
\end{defn}
\noindent From \eqref{eq:Fofxk}, the Markov jump system \mbox{\small $\Sigma$} can be written as follows: 
\begin{align}\label{eq:Fofxk_MarkovJump}
x_{k+1} &= F(x_k, H^{-1}(\eta(k))),\quad k\in\NN,\nonumber\\
x_0 &= \xi_0,\quad \xi_0\in\mathcal{C},
\end{align}
\noindent where \mbox{\small $x_k\in\mathcal{C},\,x_k(\theta)=x(k+\theta),\,\theta\in\{-\Delta,-\Delta+1,\ldots,0\},\,k\in\NN;$} \mbox{\small $\eta(k)\in\mathcal{S},$} \mbox{\small $k\in\NN$} is a MC with TPM \mbox{\small $P$} defined by \eqref{eq:TPMdef}. 
The map \mbox{\small $F$} is defined by \eqref{eq:Fofphid} and it can be rewritten as follows for \mbox{\small $\phi\in\mathcal{C},$} and for  \mbox{\small $i\in\mathcal{S}:$}
\begin{align}\label{eq:FofphiH1ofi}
    &F(\phi, H^{-1}(i)) (\theta) =\nonumber\\
    &=\bc
    f(\phi(0),\phi(-H_1^{-1}(i)),\ldots, \phi(-H_r^{-1}(i))),\,  \theta=0,\\
    \phi(\theta+1),\, \theta = -\Delta, -\Delta+1, \ldots, -1,
    \ec
\end{align}
\noindent with 
\mbox{\small$H^{-1}(i)=\bmat H^{-1}_1(i) & \ldots & H^{-1}_r(i)\emat^T\in D,\,\forall i\in\mathcal{S}.$}\\
\noindent Let \mbox{\small$x_k(\xi_0),$} \mbox{\small$k\in\NN,$} denote the trajectory that evolves according to \eqref{eq:Fofxk_MarkovJump}, corresponding to initial state \mbox{\small$\xi_0\in\mathcal{C}.$} Recall that \mbox{\small$x_k(\xi_0)(0)=x(k,\xi_0),$} \mbox{\small$k\in\NN.$}\\\noindent 
In the following example, we show step-by-step the methodology leading system \eqref{eq:plain time delay system} to system  \eqref{eq:Fofxk_MarkovJump}.
\begin{myeg}\label{ex:sat}
\noindent Consider the following scalar nonlinear system (see \cite[{\it Example 2}]{Pepe2020}) described by the following equation, for \mbox{\small$k\in\NN:$}
\beq\label{eq:ex1_1}
\bc
x(k+1)&=\mathrm{sat}(x(k))-\gamma \mathrm{sat}(x(k-d(k))),\\
x(\tau)&=\xi_0(\tau),\,\tau\in\{-2,-1,0\},
\ec
\eeq
\noindent with \mbox{\small$\xi_0\in\mathcal{C},$} \mbox{$x(k)\in\RR,$} \mbox{$\gamma\in [1,1.2].$} The delay \mbox{\small $d(k)\in D=\{0,2\}.$} The system described by \eqref{eq:ex1_1} is stable for \mbox{\small$d(k)=0,$} and it is unstable for \mbox{\small$d(k)=2$} for all \mbox{\small$k\in\NN.$}\\
\noindent We transform system \eqref{eq:ex1_1} to a system defined in the space of initial conditions using equation \eqref{eq:Fofxk}:
\begin{align}\label{eq:F_ex}
    F(\phi, d ) (\theta) =
    \bc
    \mathrm{sat}(\phi(0))-\gamma \mathrm{sat}(\phi(-d)),\,  \theta=0,\\
    \phi(\theta+1),\, \theta = -2, -1.
    \ec
\end{align}
Consider the MC \mbox{\small $\{\eta_k\}_{k\in\NN}$} in Figure \ref{fig:MC_ex1}, with set of states \mbox{\small $\mathcal{S}=\{1,2\}.$}
We associate each delay in \mbox{\small $D$} to a mode of the Markov chain in \mbox{\small $\mathcal{S}.$}
\begin{figure}
\centering
\begin{tikzpicture}[->,node distance=3cm, auto]
\node[circle, draw=white](0)at (-1,1.5){$\eta(k):$};
\node[circle, draw=black](1)at (0,0){$1$};
\node[circle, draw=black](2)at (3,0){$2$};
\path (1)edge[loop above]node{$p$}(1);
\path (1)edge[bend left]node{$1-p$}(2);
\path (2)edge[bend left]node{$1-q$}(1);
\path (2)edge[loop above]node{$q$}(2);
\end{tikzpicture}	
\caption{The Figure depicts the state diagram of the Markov chain \mbox{\small$\eta(k)$} modeling the switching delay in the presented example: \mbox{\small$p$} stands for the probability of having a delay \mbox{\small$d(k+1)=0$} provided that the previous delay is \mbox{\small$d(k)=0,$} while \mbox{\small$q$} stands for the probability of having a delay \mbox{\small$d(k+1)=2,$} provided that the previous delay is \mbox{\small$d(k)=2.$}}\label{fig:MC_ex1}
\end{figure}
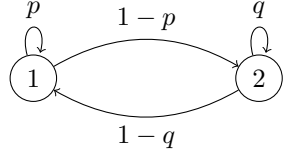
Let the bijective function \mbox{\small$H:D\to\mathcal{S}$} be defined as
\beq\label{eq:ex1_2}
H(d)=\bc
1 & \text{if }d=0,\\
2 & \text{if }d=2. 
\ec
\eeq
\noindent The function \mbox{\small$H^{-1}:\mathcal{S}\to D$} is thus defined as 
\beq\label{eq:ex1_3}
H^{-1}(i)=\bc 
0 & \text{if } i=1,\\
2 & \text{if } i=2.
\ec
\eeq
\noindent The TPM associated with the MC \mbox{\small$\{\eta_k\}_{k\in\NN}$} is given by
\beq\label{eq:ex1_3_1}
P=\bmat p & 1-p\\
      1-q & q\emat,\quad p,q\in (0,1).
\eeq
Hence, we obtain a Markov jump system \mbox{\small$\Sigma$} where \mbox{\small$\mathcal{D}$} is \eqref{eq:ex1_1}, \mbox{\small$P$} is defined in \eqref{eq:ex1_3_1} and \mbox{\small$H$} is defined in \eqref{eq:ex1_2}. 
\end{myeg}
\noindent\begin{rmk}\label{rem:xofk_random_variable}
Notice that the variable \mbox{\small$x(k,\xi_0)\in\RR^n,$} \mbox{\small$\xi_0\in\mathcal{C},$} \mbox{\small$k\in\NN,$} is a random variable on the stochastic basis \mbox{\small$\left(\Omega,\mathcal{G},\mathcal{G}_k,\mathbb{P}\right),$} since the delay evolves according to a discrete-time MC, with given transition probabilities. Thus, we are interested in the behaviour of the second moment of \mbox{\small$x(k,\xi_0),$} \mbox{\small$k\in\NN,$} \mbox{\small$\xi_0\in\mathcal{C}.$} 
\end{rmk}
\begin{defn}\label{def:EMSS}
	The Markov jump system \mbox{\small$\Sigma$} is {\it EMSS}
	if there exist \mbox{\small$M,\,\zeta \in \RR^{+}$} with \mbox{\small$M \geq 1$} and \mbox{\small$0< \zeta < 1,$} such that for any \mbox{\small$\xi_{0} \in \mathcal{C},$}  the following inequality holds for any \mbox{\small$k\in\NN,$}
\begin{align}
	    \EE[ \lVert x(k, \xi_{0})\rVert^{2}] \leq M \zeta^k \Big( \lVert \xi_{0}\rVert_{\infty}\Big )^{2}.
\end{align}
\end{defn}
\section{Main result}
\noindent In this section, we provide the main result of the paper. We derive sufficient Lyapunov conditions guaranteeing the {\it EMSS} of system \mbox{\small $\Sigma.$}\\
\noindent Let us consider a scalar function \mbox{\small$V :\mathcal{C}\times D \to \RR^{+}.$} Let us associate to \mbox{\small$V$} the operator \mbox{\small$\mathcal{L}V:\mathcal{C}\times D \to \RR,$} defined for \mbox{\small$\phi\in\mathcal{C},$} \mbox{\small$i\in\mathcal{S},$} as
\begin{align}\label{eq:mathcalLVdef}
&\mathcal{L} V (\phi, H^{-1}(i)) \triangleq\nonumber\\
		&{\displaystyle \sum_{j\in \mathcal{S}}} p_{ij} V\Big( F(\phi,H^{-1}(i) ), H^{-1}(j)\Big )- V(\phi, H^{-1}(i)),
\end{align}
\noindent with \mbox{\small$H^{-1}$} defined in \eqref{eq:bijective_inverseH}.
\begin{thm}\label{thm:msscomplete}
\noindent Assume there exist a function 
\mbox{$V :\mathcal{C}\times D \to \RR^{+}$} and real positive numbers \mbox{\small $\alpha_{i},$} \mbox{\small$i=1,2,3,$} such that, for all \mbox{\small$\phi \in \mathcal{C},$} for all \mbox{\small$i\in\mathcal{S},$} 
the following inequalities hold:
\begin{itemize}
    \item[$i)$] $\alpha_{1}\lVert \phi(0) \rVert ^2 \leq V(\phi, H^{-1}(i)) \leq \alpha_{2}\lVert \phi \rVert_{\infty} ^2,$
    \item[$ii)$]  
	$\mathcal{L} V (\phi, H^{-1}(i)) \leq - \alpha_3 \lVert \phi(0) \rVert ^2,$
\end{itemize}
where {\small $H^{-1}$} is defined in \eqref{eq:bijective_inverseH}.\\
\noindent Then, the system \mbox{\small$\Sigma$} is {\it EMSS}.
\end{thm}
\begin{proof}
See the Appendix.
\end{proof}
\begin{rmk}\label{rem:emphasis_our_conditions}
Notice that the classical representation of system $\Sigma$ using a Markov jump system allows to use the conditions in \cite{Impicciatore2020sufficient,Patrinos2014stochastic}.
An important feature of the involved Lyapunov inequalities, shared with the cases of delay-dependent and delay-independent Lyapunov functions (see \cite{Pepe2020,Pepe2018OnLyap}), is that lower bound of Lyapunov functions, as well as of the related difference
operators, are given in a weaker form with respect to the Lyapunov conditions for Markov jump systems (see for instance \cite{Impicciatore2020sufficient,Patrinos2014stochastic})
Indeed, the lower bound of condition $(i)$ in Theorem \ref{thm:msscomplete} and the inequality in condition $(ii)$ in Theorem \ref{thm:msscomplete} do not involve $\phi\in\mathcal{C}$, but only $\phi(0)\in\RR^n$.
\end{rmk}
\section{EXAMPLE}
In this section, we aim to study the {\it EMSS} property of the system \mbox{\small$\Sigma$} obtained in Example \ref{ex:sat} starting from the discrete-time delay system \eqref{eq:ex1_1}.\\
We analyze the Markov jump system resulting by the application of our methodology. As mentioned before, we obtain a switching system with two modes: one is stable and the other one is unstable. \\
\noindent Notice that, as consequence of the structure of \mbox{\small $P$} in \eqref{eq:ex1_3_1}, the set \mbox{\small$\mathcal{E}(D)$} is given by
\mbox{\small$\mathcal{E}(D)=\{(0,0),(0,2),(2,0),(2,2)\}.$}
\noindent In the following, we want to verify whether conditions $(i)$ and $(ii)$ of Theorem \ref{thm:msscomplete} are satisfied.\\
\noindent We consider a candidate Lyapunov function \mbox{\small$V:\mathcal{C}\times D\to\RR^+$} 
defined, for \mbox{\small$\phi \in \mathcal{C},$} \mbox{\small$i\in \mathcal{S},$} as 
\beq\label{eq:ex1_4}
V(\phi,H^{-1}(i))=\lambda_i{\displaystyle\sup_{j=0,1,2}} 2^{j-1}\gamma^j c^{-j}\|\phi(-j)\|^2,
\eeq
\noindent with \mbox{\small$\lambda_i\in\RR^+,$} \mbox{\small$i\in\mathcal{S},$} \mbox{\small$\gamma\in [1, 1.2],$} \mbox{\small$1< c \leq e.$}\\
\noindent Pick \mbox{\small$\alpha_1={\displaystyle\min_{i\in\mathcal{S}}}\lambda_i,\,\alpha_2= 2\gamma^2 {\displaystyle\max_{i\in\mathcal{S}}}\lambda_i.$}\\ 
\noindent Thus, condition $(i)$ of Theorem \ref{thm:msscomplete} is satisfied. In order to verify condition $(ii)$, we consider the expression of \mbox{\small$\mathcal{L}V(\phi,H^{-1}(i)),$} for all \mbox{\small$\phi\in\mathcal{C},$} \mbox{\small$i\in\mathcal{S}.$} \\\noindent When we consider the expression of \mbox{\small$\mathcal{L}V(\phi,H^{-1}(1)),$} from \eqref{eq:mathcalLVdef} we obtain the following equality:
\begin{align}\label{eq:ex1_6_0}
    &\mathcal{L}V(\phi,H^{-1}(1))\nonumber\\
&=p V\left(F(\phi,H^{-1}(1)),H^{-1}(1)\right)
\nonumber\\
& +(1-p)V\left(F(\phi,H^{-1}(1)),H^{-1}(2)\right)- V(\phi,H^{-1}(1)).
\end{align}
\noindent From \eqref{eq:ex1_4} and \eqref{eq:ex1_6_0}, we obtain the following equality:
\begin{align}\label{eq:ex1_6_1}
&\mathcal{L}V(\phi,H^{-1}(1))=\nonumber\\
&=\!(p\lambda_1\!+\!(1\!-\!p)\lambda_2)\!\!\sup_{j=0,1,2}2^{j-1}\gamma^j c^{-j}\|F(\phi,H^{-1}(1)\!)\!(-j)\|^2
\nonumber\\
&-\lambda_1\sup_{j=0,1,2}2^{j-1}\gamma^j c^{-j}\|\phi(-j)\|^2.
\end{align}
By \eqref{eq:ex1_3}, we have:
\begin{align}\label{eq:exmy6}
&\mathcal{L}V(\phi,H^{-1}(1))\nonumber\\
&\leq (p\lambda_1+(1-p)\lambda_2)2^{-1}\lVert(1-\gamma)\mathrm{sat}(\phi(0))\rVert^2
\nonumber\\
&+ (p\lambda_1+(1-p)\lambda_2)\sup_{j=1,2}2^{j-1}\gamma^j c^{-j}\|\phi(-j+1)\|^2\nonumber\\
&-\lambda_1\sup_{j=0,1,2}2^{j-1}\gamma^j c^{-j}\lVert \phi(-j)\rVert^2.
\end{align}
\noindent By the properties of the supremum, the following inequality holds
\begin{align}\label{eq:ex_1_6_2}
&\sup_{j=1,2}2^{j-1}\gamma^j c^{-j}\|\phi(-j+1)\|^2\leq\nonumber\\ &\sup_{j=1,2,3}2^{j-1}\gamma^j c^{-j}\|\phi(-j+1)\|^2.
\end{align}
\noindent By changing the index variable in the supremum, we can write
\begin{align}\label{eq:ex_1_6_3}
    &\sup_{j=1,2,3}2^{j-1}\gamma^j c^{-j}\|\phi(-j+1)\|^2=\nonumber\\
    &2\gamma c^{-1}\sup_{j=1,2,3}2^{(j-1)-1}\gamma^{(j-1)} c^{-(j-1)}\|\phi(-j+1)\|^2=\nonumber\\
    &2\gamma c^{-1}\sup_{\theta=0,1,2}2^{\theta -1}\gamma^\theta c^{-\theta}\lVert\phi(-\theta)\rVert^2.
\end{align}
\noindent Thus, from \eqref{eq:exmy6}, \eqref{eq:ex_1_6_2}, \eqref{eq:ex_1_6_3}, we obtain the following inequalities
\begin{align}\label{eq:ex1_6_4}
&\mathcal{L}V(\phi,H^{-1}(1))\leq (p\lambda_1 +(1-p)\lambda_2)\Big(2^{-1}(1-\gamma)^2\lVert\phi(0)\rVert^2+\nonumber\\
&+2\gamma c^{-1}\sup_{\theta=0,1,2}2^{\theta-1}\gamma^{\theta} c^{-\theta}\lVert \phi(-\theta)\rVert^2\Big)+\nonumber\\
&-\lambda_1\sup_{\theta=0,1,2}2^{\theta-1}\gamma^{\theta} c^{-\theta}\lVert \phi(-\theta)\rVert^2\nonumber\\
&\leq (p\lambda_1 +(1-p)\lambda_2)\left((1-\gamma)^2+2\gamma c^{-1}\right)\nonumber\\
&\times\sup_{\theta=0,1,2}2^{\theta-1}\gamma^{\theta} c^{-\theta}\lVert \phi(-\theta)\rVert^2\nonumber\\
&-\lambda_1\sup_{\theta=0,1,2}2^{\theta-1}\gamma^{\theta} c^{-\theta}\lVert \phi(-\theta)\rVert^2.
\end{align}
\noindent By defining \mbox{\small$\omega_1$} as follows,
\beq\label{eq:ex1_6_4_1}
\omega_1\triangleq\lambda_1\left[1-\left(p+(1-p)\dfrac{\lambda_2}{\lambda_1}\right)\left((1\!-\!\gamma)^2\!+\!2\gamma c^{-1}\right)\right],
\eeq
\noindent we get
\beq\label{eq:ex1_6_5}
\mathcal{L}V(\phi,H^{-1}(1))\leq -\omega_1{\displaystyle\sup_{\theta=0,1,2}} 2^{\theta-1}\gamma^{\theta} c^{-\theta}\lVert \phi(-\theta)\rVert^2.
\eeq
\noindent When we consider the expression of \mbox{\small$\mathcal{L}V(\phi,H^{-1}(2)),$} we obtain:
\begin{align}\label{eq:ex1_6_6}
&\mathcal{L}V(\phi,H^{-1}(2))=(1-q)V(F(\phi,H^{-1}(2)),H^{-1}(1)) \nonumber\\
&+ q V(F(\phi,H^{-1}(2)),H^{-1}(2))-V(\phi,H^{-1}(2)).
\end{align}
\noindent From \eqref{eq:ex1_6_6}, and \eqref{eq:ex1_4}, the following equality holds:
\begin{align}\label{eq:ex1_6_7}
&\mathcal{L}V(\phi,H^{-1}(2))=\nonumber\\
&=\left((1-q)\lambda_1 + q\lambda_2\right)\!\!\sup_{j=0,1,2}\!2^{j-1}\gamma^j c^{-j}\lVert F(\phi,H^{-1}(2))(-j)\!\rVert^2\nonumber\\
&-\lambda_2\sup_{j=0,1,2} 2^{j-1}\gamma^j c^{-j}\lVert \phi(-j)\rVert^2.
\end{align}
\noindent From \eqref{eq:ex1_6_7},  applying the properties of the supremum, it follows that
\begin{align}\label{eq:ex1_6_8}
&\mathcal{L}V(\phi,H^{-1}(2))\nonumber\\ 
&\leq\left((1-q)\lambda_1 +q \lambda_2\right)\Big( 2^{-1}\lVert F(\phi,H^{-1}(2))(0)\rVert^2 \nonumber\\
&+\sup_{j=1,2} 2^{j-1}\gamma^j c^{-j}\lVert \phi(-j+1)\rVert^2
\Big)\nonumber\\
&
-\lambda_2\sup_{j=0,1,2}2^{j-1}\gamma^j c^{-j}\lVert\phi(-j)\rVert^2.
\end{align}
\noindent From \eqref{eq:ex1_6_8}, we have
\begin{align}\label{eq:ex1_6_9}
&\mathcal{L}V(\phi,H^{-1}(2))\nonumber\\
&\leq\left((1-q)\lambda_1 +q\lambda_2\right)\Big( 2^{-1}
\lVert \mathrm{sat}(\phi(0))-\gamma \mathrm{sat}(\phi(-2))\rVert^2 +\nonumber\\
&+2\gamma c^{-1}\sup_{j=1,2} 2^{(j-1)-1}\gamma^{j-1} c^{-j+1}\Vert\phi(-j+1)\rVert^2\Big)\nonumber\\
&-\lambda_2\sup_{\theta=0,1,2} 2^{\theta-1}\gamma^{\theta}c^{-\theta}\lVert\phi(-\theta)\rVert^2.
\end{align}
\noindent From \eqref{eq:ex1_6_9}, by applying the properties of the Euclidean norm, Young's inequality and the properties of the function $\mathrm{sat}$ the following inequalities hold
\begin{align}\label{eq:ex1_6_10}
  &\mathcal{L}V(\phi,H^{-1}(2))\nonumber\\  &\leq\left((1-q)\lambda_1+q\lambda_2\right)\Big(\lVert\phi(0)\rVert^2+\gamma^2\lVert 
\phi(-2)\rVert^2\nonumber\\
&+2\gamma c^{-1}\sup_{\theta=0,1,2} 2^{\theta-1}\gamma^{\theta} c^{-\theta}\lVert\phi(-\theta)\rVert^2\Big)\nonumber\\
&-\lambda_2\sup_{\theta=0,1,2} 2^{\theta-1}\gamma^{\theta}c^{-\theta}\lVert \phi(-\theta)\rVert^2\nonumber\\
&\leq\left((1-q)\lambda_1+q\lambda_2\right)\Big(\left(2+2^{-1}c^2+2\gamma c^{-1}\right)\nonumber\\
&\times\sup_{\theta=0,1,2}2^{\theta -1}\gamma^{\theta}c^{-\theta}\lVert\phi(-\theta)\rVert
^2\Big)\nonumber\\
&-\lambda_2
\sup_{\theta=0,1,2}2^{\theta-1}\gamma^{\theta}c^{-\theta}\lVert\phi(-\theta)\rVert^2.
\end{align}
\noindent From \eqref{eq:ex1_6_10}, by defining \mbox{\small$\omega_2$} as follows,
\beq\label{eq:ex1_6_11}
\omega_2\triangleq\lambda_2\left[1-\left(q+(1-q)\dfrac{\lambda_1}{\lambda_2}\right)\left(2+2^{-1}c^2+2\gamma c^{-1}\right)\right],
\eeq
\noindent we obtain
\begin{align}\label{eq:ex1_6_12}
\mathcal{L}V(\phi,H^{-1}(2))\leq -\omega_2{\displaystyle\sup_{\theta =0,1,2}} 2^{\theta -1}\gamma^{\theta}c^{-\theta}\lVert \phi(-\theta)\rVert ^2.   
\end{align}
\noindent Under the following constraints
\bseq\label{eq:ex1_7}
\beq\label{eq:ex1_7_1}
L_B<\dfrac{\lambda_2}{\lambda_1}<U_B,
\eeq
\noindent with
\beq\label{eq:ex1_7_2}
U_B=\dfrac{1-((1-\gamma)^2+2\gamma c^{-1})p}{((1-\gamma)^2+2\gamma c^{-1})(1-p)},
\eeq
\beq\label{eq:ex1_7_3}
L_B=\dfrac{\left(4+c^{2}+4\gamma c^{-1}\right)(1-q)}{2-\left(4+c^{2}+4\gamma c^{-1}\right)q},
\eeq
\beq\label{eq:ex1_7_4}
(p,q)\in (0,1)\times (0,1),\, q<\dfrac{2}{4+c^{2}+4\gamma c^{-1}};
\eeq
\eseq
\noindent we obtain that \mbox{\small$\omega_1,\,\omega_2\in\RR^+.$} Thus, from \eqref{eq:ex1_6_5} and \eqref{eq:ex1_6_12} the following inequality holds, for all \mbox{\small$\phi\in\mathcal{C},$} for all \mbox{\small$i\in\mathcal{S},$}
\beq\label{eq:ex1_8}
\mathcal{L}V\left(\phi,H^{-1}(i)\right)\leq -\alpha_3\lVert\phi(0)\rVert^2,
\eeq
\noindent with \mbox{\small$\alpha_3\in\RR^+,$} defined as \mbox{\small $\alpha_3\triangleq\frac{1}{2}\min\{\omega_1,\omega_2\}.$}
\noindent Thus, condition $(ii)$ of Theorem \ref{thm:msscomplete} is satisfied and the system \eqref{eq:ex1_1} with 
switches delays governed by the MC 
\mbox{\small$\{\eta(k)\}_{k\in\NN},$} with TPM P defined in \eqref{eq:ex1_3_1}, and with the function \mbox{\small$H$} defined in \eqref{eq:ex1_3} is {\it EMSS}.
\begin{rmk}\label{rem:lastrem}
Notice that, the conditions for the exponential mean square stability of Markov jump systems in \cite[Theorem 20(b)]{Patrinos2014stochastic} (see also \cite{Impicciatore2020sufficient})  
cannot be applied in the example \eqref{eq:ex1_1}
on the extended state in \mbox{\small$\RR^{n(\Delta +1)}.$}
\end{rmk}
\subsection{Statistical results} 
\noindent In Figure \ref{fig:Sim1_3stable}, we present Montecarlo simulations of
the trajectories generated by the system \eqref{eq:ex1_1},  considering values of the pairs \mbox{\small$(p,q)$} such that conditions $(i)$ and $(ii)$ of Theorem \ref{thm:msscomplete} are satisfied.
The yellow trajectories correspond to the state trajectories associated with different switching paths (that are admissible according to \mbox{\small$P$}), the maximum and the minimum trajectory are plotted in blue and green, respectively. Finally, the red line corresponds to the average evolution of the state trajectories. From Figure \ref{fig:Sim1_3stable}, we observe that trajectories decrease exponentially and converge to zero. This result reflects the analysis presented in this section.
\begin{figure}[h]
      \centering
\includegraphics[scale=0.19]{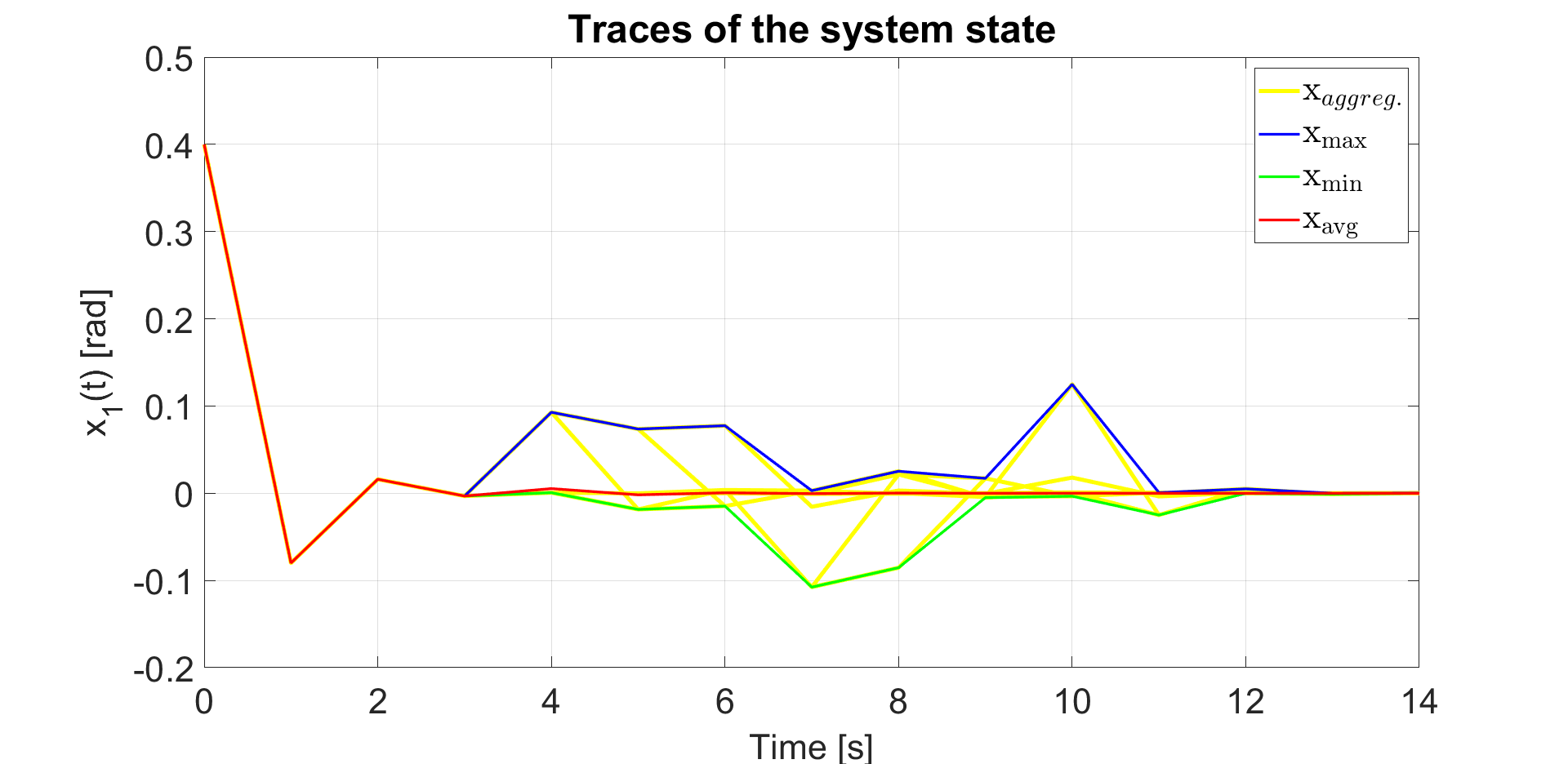}\caption{
Traces of system state obtained with \mbox{\small $\gamma=1.2,$} \mbox{\small $p=0.95,$} and \mbox{\small $q=0.01.$}}
      \label{fig:Sim1_3stable}
\end{figure}
\begin{figure*}
      \centering
\subcaptionbox{$c=5.2,\gamma=1$\label{a}} {\includegraphics[scale=0.15]{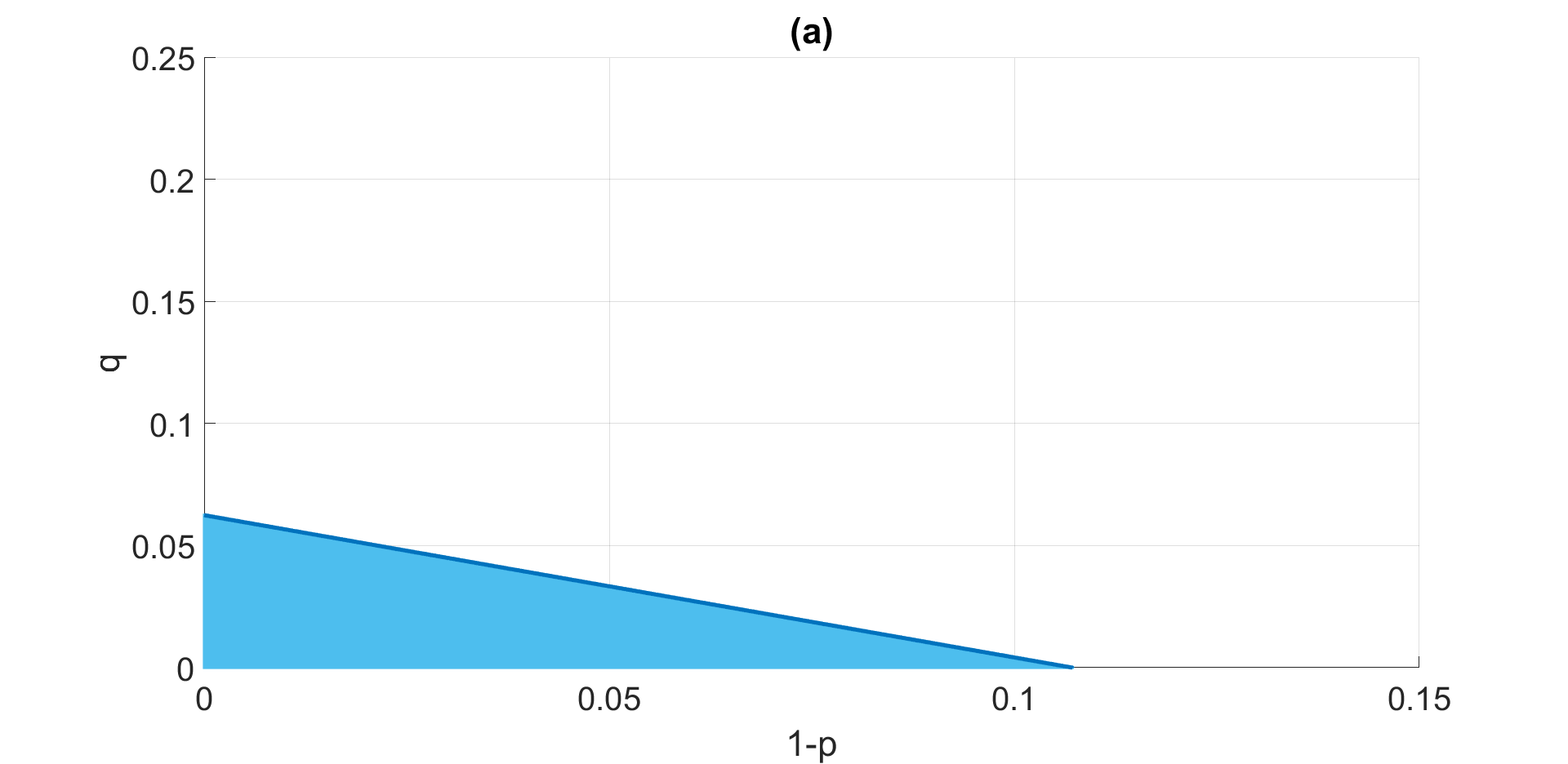}}\subcaptionbox{$c=5.2,\gamma=1.2$\label{b}}{\includegraphics[scale=0.15]{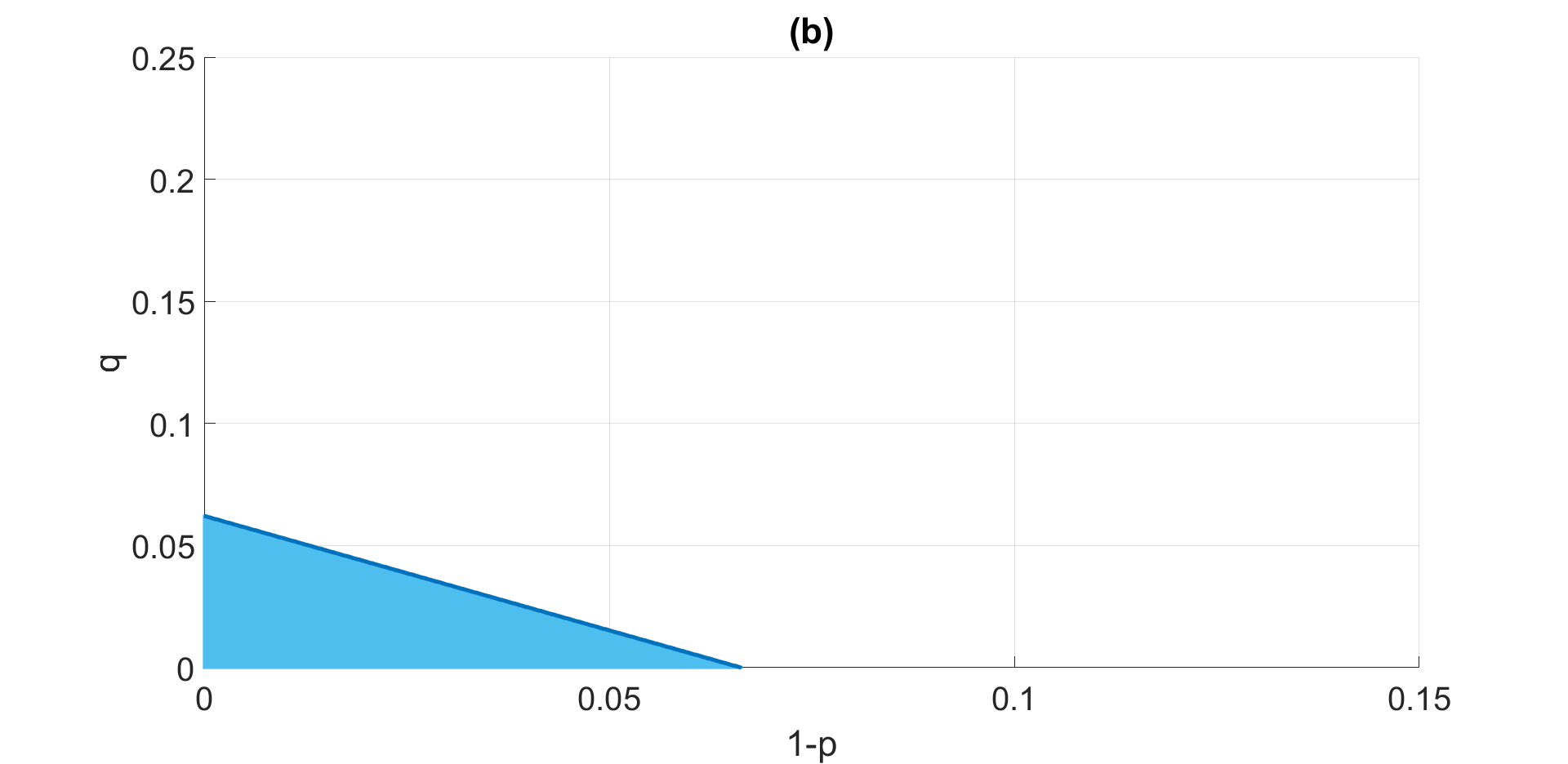}}\\
\subcaptionbox{$c=e,\gamma=1$\label{c}}{\includegraphics[scale=0.15]{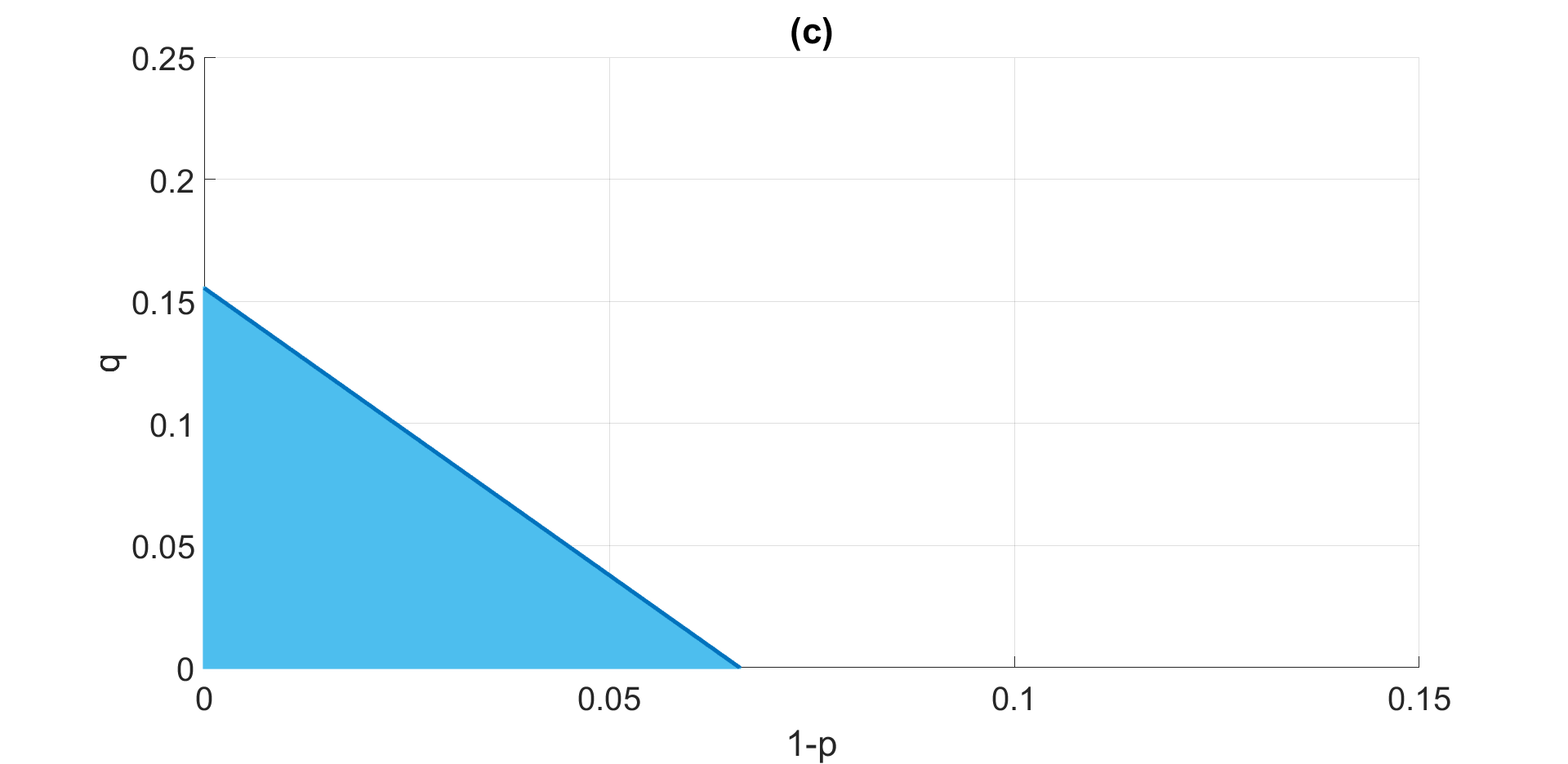}} \subcaptionbox{$c=e,\gamma=1.2$ \label{d}}{\includegraphics[scale=0.15]{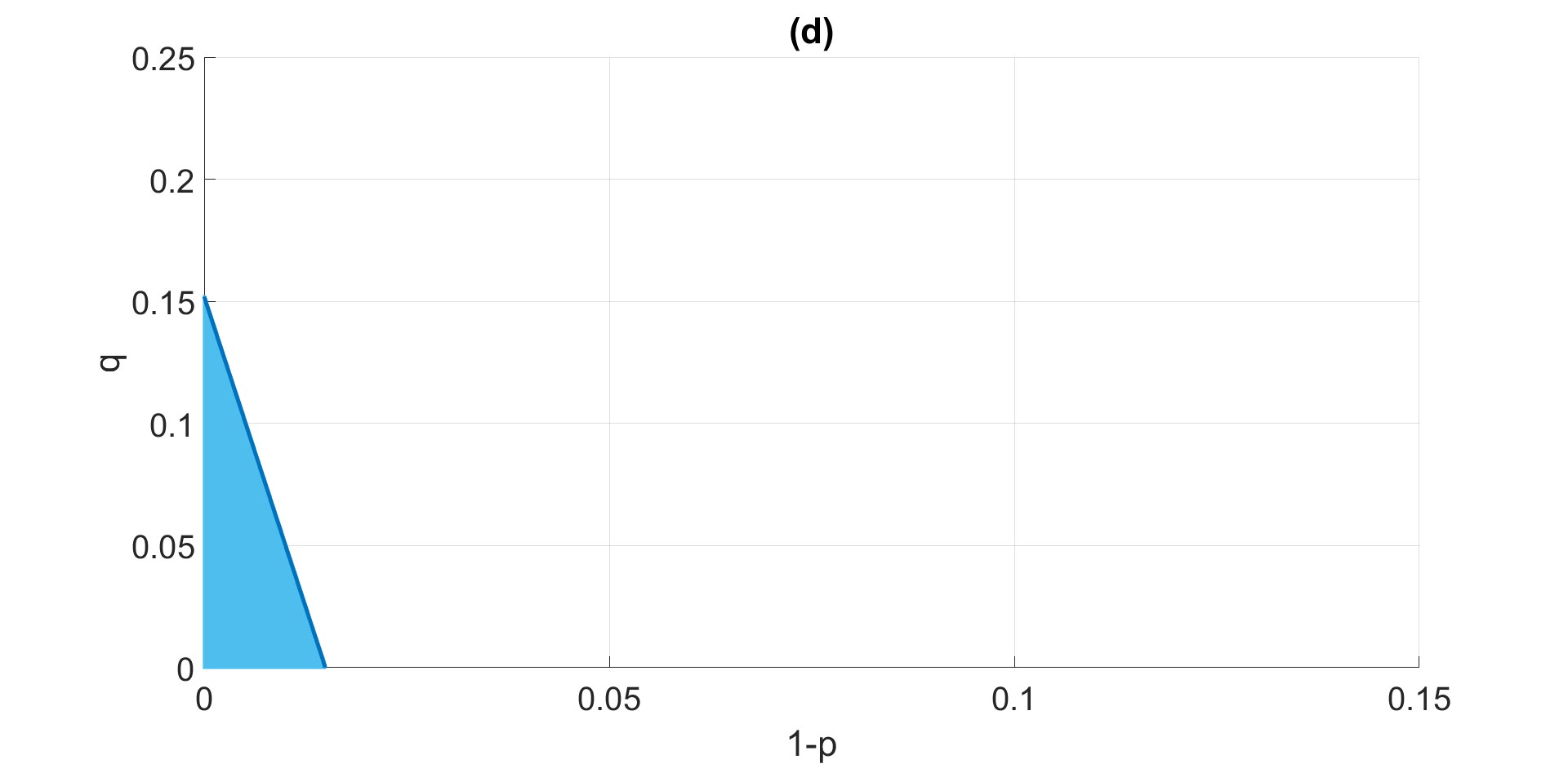}}
\caption{The figure shows the regions of pairs \mbox{\small $(p,q)\in (0,1)\times (0,1)$} such that conditions of Theorem \ref{thm:msscomplete} are satisfied .\label{fig:plotpandq}} 
\end{figure*}
In Figures \ref{fig:plotpandq},
for different values of \mbox{\small$c$} in the candidate Lyapunov function \eqref{eq:ex1_4},
we show the regions of 
 pairs \mbox{\small$(p,q)$} such that conditions $(i)-(ii)$ of Theorem \ref{thm:msscomplete} are satisfied (light blue region) and the evolution of the maximum \mbox{\small$q$} with respect to \mbox{\small$(1-p)$} such that the conditions $(i)-(ii)$ of Theorem \ref{thm:msscomplete} are satisfied (dark blue line).
From Figures \ref{fig:plotpandq},
by comparing row-wise, we observe that
when \mbox{\small$c$} spans from \mbox{\small$e$} to \mbox{\small$5.2,$} the segment on \mbox{\small$1-p$} shrinks while
segment on \mbox{\small$q$} expands.
From Figures \ref{fig:plotpandq}, by comparing column-wise,
when \mbox{\small$\gamma$} goes from \mbox{\small$1$} to \mbox{\small$1.2,$} the segment on $1-p$ shrinks, the aforementioned light blue region becomes smaller and smaller when the parameter \mbox{\small$\gamma$} increases.
Choosing \mbox{\small$c>e,$} the conditions 
\eqref{eq:ex1_7} lead to a widen set of values for \mbox{\small$1-p,$}
while, by condition \eqref{eq:ex1_7_4}, the values of \mbox{\small$q$} are restricted.
\section{CONCLUSIONS}\label{sec:conclusion}
\noindent In this paper, we 
provide sufficient Lyapunov conditions guaranteeing the {\it EMSS} property of discrete-time systems with markovian delays.
\noindent Future work directions would be the extension of sufficient conditions for the exponential input-to-state stability in mean square sense, as well as necessary and  sufficient conditions guaranteeing {\it EMSS} for this class of systems.

\bibliography{bib_MED}
\section*{APPENDIX}
\subsection*{A. Preliminary Results Needed for the Proofs of Theorem \ref{thm:msscomplete}}
\noindent In the following, we introduce some technical results that are useful in the proof of Theorem \ref{thm:msscomplete}.
\begin{lem}\label{lem:1}
Let there exist a function $V :\mathcal{C}\times D \to \RR^{+}$, real positive numbers $\alpha_{i}$, $i=1,2,3$, such that, for all $\phi \in \mathcal{C},$ for all $i \in \mathcal{S}$,
the following inequalities hold:
\begin{itemize}
    \item[$a_1)$] $\alpha_{1}\lVert \phi(0) \rVert ^2 \leq V(\phi, H^{-1}(i)) \leq \alpha_{2}\lVert \phi \rVert_{\infty} ^2$,
    \item[$a_2)$] $\mathcal{L} V (\phi, H^{-1}(i))\leq - \alpha_3 \lVert \phi(0) \rVert ^2,$
\end{itemize}
\noindent with the operator $\mathcal{L}V$ defined in \eqref{eq:mathcalLVdef}.
\noindent Then, there exist a function $W :\mathcal{C}\times D \to \RR^{+}$,
real positive numbers $\beta_{i}$, $i=1,2,3$, such that, for all $\phi \in \mathcal{C},$ for all $i \in \mathcal{S}$,
the following inequalities hold:
\begin{itemize}
    \item[$b_1)$] $\beta_{1}\lVert \phi(0) \rVert ^2 \leq W(\phi, H^{-1}(i)) \leq \beta_{2}\lVert \phi \rVert_{\infty} ^2$,
    \item[$b_2)$] $	{\displaystyle \sum_{j\in \mathcal{S}}} p_{ij} W\Big( F(\phi,H^{-1}(i) ), H^{-1}(j)\Big )- W(\phi, H^{-1}(i)) 
		\leq - \beta_3 \lVert \phi \rVert_{\infty}^{2}$,
\end{itemize}
\noindent where $H^{-1}$ is defined in \eqref{eq:bijective_inverseH}.
\end{lem}
\begin{proof}
Let us consider the function $W :\mathcal{C}\times D \to \RR^{+}$ defined, for $\phi \in \mathcal{C},$ $i \in \mathcal{S}$, as (we are inspired by \cite[Lemma 1]{Pepe2018OnLyap}, \cite[Lemma 1]{Pepe2020} which deal with the deterministic case):
\beq
W(\phi,H^{-1}(i) ) = V(\phi,H^{-1}(i)) + \max_{\theta=1,2,\ldots, \Delta} e^{-\theta} \alpha_3 \lVert \phi(-\theta) \rVert^{2}.
\eeq
Then, from $(a_1)$, we have
\begin{align}
    \beta_{1}\lVert \phi(0) \rVert ^2 \leq W(\phi, H^{-1}(i))  \leq \beta_{2}\lVert \phi \rVert_{\infty} ^2,
\end{align}
\noindent with $\beta_1=\alpha_1$ and $\beta_2 = \alpha_2 + \alpha_3$.\\
\noindent From $(a_2)$, we obtain
\begin{align*}
    &\sum_{j\in \mathcal{S}} p_{ij} W\Big( F(\phi,H^{-1}(i)), H^{-1}(j)\Big) - W(\phi, H^{-1}(i))  =\\
    &\sum_{j\in \mathcal{S}} p_{ij} V\Big( F(\phi,H^{-1}(i)), H^{-1}(j)\Big) - V(\phi,H^{-1}(i))\\
    &+ \max_{\theta=1,2,\ldots, \Delta} e^{-\theta} \alpha_3 \lVert \phi(-\theta+1) \rVert^{2}\\
    &- \max_{\theta=1,2,\ldots, \Delta} e^{-\theta} \alpha_3 \lVert \phi(-\theta) \rVert^{2}\\
    &\leq  - \alpha_3 \lVert \phi(0) \rVert ^2\\
    &+ e^{-1} \max_{\theta=1,2,\ldots, \Delta} e^{1-\theta} \alpha_3 \lVert \phi(-\theta+1) \rVert^{2}\\
    &- \max_{\theta=1,2,\ldots, \Delta} e^{-\theta} \alpha_3 \lVert \phi(-\theta) \rVert^{2}\\
    &\leq - \alpha_3 \lVert \phi(0) \rVert ^2\\
    &+ e^{-1} \max_{\theta=0,1,\ldots, \Delta-1} e^{-\theta} \alpha_3 \lVert \phi(-\theta) \rVert^{2}\\
    &- \max_{\theta=1,2,\ldots, \Delta} e^{-\theta} \alpha_3 \lVert \phi(-\theta) \rVert^{2}\\
    &\leq - \alpha_3 \lVert \phi(0) \rVert ^2 +  e^{-1} \alpha_3 \lVert \phi(0) \rVert ^2\\
    &+ e^{-1} \max_{\theta=1,2,\ldots, \Delta} e^{-\theta} \alpha_3 \lVert \phi(-\theta) \rVert^{2}\\
    &- \max_{\theta=1,2,\ldots, \Delta} e^{-\theta} \alpha_3 \lVert \phi(-\theta) \rVert^{2}\\
    &\leq - ( 1- e^{-1}) \alpha_3 \lVert \phi(0) \rVert ^2\\
    & -( 1- e^{-1}) \alpha_3 e^{-\Delta} \max_{\theta=1,2,\ldots, \Delta}\lVert \phi(-\theta) \rVert^{2}\\
    &\leq -( 1- e^{-1}) \alpha_3 e^{-\Delta}\lVert \phi\rVert_{\infty}^{2}.
\end{align*}
\noindent Define $\beta_3 \triangleq ( 1- e^{-1}) \alpha_3 e^{-\Delta}$. Then, the function $W$ satisfies $(b_1),$ $(b_2)$. This completes the proof. 
\end{proof}
\begin{lem}\label{lem:mssnew}
	Assume that there exist a function $V:\mathcal{C}\times D \to \RR^{+}$, real positive numbers $\gamma_{i}$, $i=1,2,3$
	such that, for all $\phi \in \mathcal{C},$ for all $i \in \mathcal{S}$,
	 the following inequalities hold:
	\begin{itemize}
		\item[$c_1)$] $\gamma_{1}\lVert \phi(0) \rVert ^2 \leq V(\phi, H^{-1}(i)) \leq \gamma_{2}\lVert \phi \rVert_{\infty} ^2$,
    \item[$c_2)$] $\mathcal{L} V (\phi, H^{-1}(i)) 
		\leq - \gamma_3 \lVert \phi \rVert_{\infty}^{2},$
	\end{itemize}
	\noindent
	where $H^{-1}$ is defined in \eqref{eq:bijective_inverseH}, and the operator $\mathcal{L}V$ is defined in \eqref{eq:mathcalLVdef}.
	Then, the system $\Sigma$ is 
	{\it EMSS}.
\end{lem}
\begin{proof}
From $(c_2)$, for all $x_k \in \mathcal{C}$, for all $\eta(k)\in\mathcal{S}$, $k\in\NN$, we have:
\begin{align}\label{eq:Lyap_1}
    &\mathcal{L} V(x_k, H^{-1}(\eta(k)))=\nonumber\\
    &=\sum_{\eta(k+1) \in \mathcal{S}} \!\!\!\!p_{\eta(k)\eta(k+1)} V(x_{k+1}, H^{-1}(\eta(k+1))) \nonumber
    \\&- V(x_k, H^{-1}(\eta(k))) \leq - \gamma_3 \lVert x_k\rVert_{\infty}^{2}.
\end{align}
\noindent By the Markov property, from \eqref{eq:Lyap_1}, we obtain:
\begin{align}\label{eq:Lyap_2}
\EE\Big[\Big(V(x_{k+1}, H^{-1}(\eta(k+1))) &- V(x_k, H^{-1}(\eta(k)))\Big) | \mathcal{G}_{k}\Big]\nonumber\\
&\leq -\gamma_3\lVert x_k\rVert_{\infty}^{2}.
\end{align}
\noindent From \eqref{eq:Lyap_2}, applying the property of the expected value conditioned to a filtration, the following inequality holds:
\begin{align}\label{eq:Lyap_3}
    \EE\Big[V(x_{k+1}, H^{-1}(\eta(k+1))) &- V(x_k, H^{-1}(\eta(k))) \Big]\nonumber\\ 
    &\leq -\gamma_3 \EE[\lVert x_k\rVert_{\infty}^{2}].
\end{align}
\noindent Using the linearity of the expected value, from \eqref{eq:Lyap_3} we obtain:
\begin{align}\label{eq:Lyap_4}
\EE[V(x_{k+1}, H^{-1}(\eta(k+1)))]&-\EE[ V(x_{k}, H^{-1}(\eta(k)))] \nonumber\\
&\leq -\gamma_3 \EE[\lVert x_k\rVert_{\infty}^{2}].
\end{align}
From $(c_1)$, it follows that
\beq\label{eq:Lyap_5}
\EE[\lVert x_k \rVert_{\infty}^2] \geq \dfrac{1}{\gamma_2 }\EE[V(x_k, H^{-1}(\eta(k))].
\eeq
Using \eqref{eq:Lyap_4} and \eqref{eq:Lyap_5}, we have
\begin{align}\label{eq:Lyap_6}
\EE[V(x_{k+1}, H^{-1}&(\eta(k+1)))]-\EE[ V(x_k, H^{-1}(\eta(k)))]\nonumber\\
&\leq -\dfrac{\gamma_3}{\gamma_2} \EE[V(x_k, H^{-1}(\eta(k)))].
\end{align}
\noindent Let $\gamma_4 \triangleq \dfrac{\gamma_3}{\gamma_2}$, notice that $\gamma_4>0$, since $\gamma_3,\,\gamma_2>0$ . 
Without loss of generality, pick $\gamma_4 <1.$
From \eqref{eq:Lyap_6}, it follows that
\beq\label{eq:Lyap_7}
\EE[V(x_{k+1}, H^{-1}(\eta(k+1)))]\leq (1 -\gamma_4) \EE[V(x_k, H^{-1}(\eta(k)))].
\eeq
Using recursive argument, from \eqref{eq:Lyap_7}, we have
\beq\label{eq:Lyap_8}
\EE[V(x_k, H^{-1}(\eta(k)))]\leq (1 -\gamma_4)^k \EE[V(\xi_0, H^{-1}(\eta(0)))].
\eeq
From $(c_1)$, for $k\in\NN$,
\begin{align}\label{eq:Lyap_9}
&\gamma_1 \EE[\lVert x(k) \rVert^2]\leq \EE[V(x_k, H^{-1}(\eta(k)))],\nonumber\\
&(1-\gamma_4)^k \EE[V(\xi_0,H^{-1}(\eta(0)))] \leq (1-\gamma_4)^k \gamma_2\EE[\lVert \xi_0\rVert_{\infty}^2].
\end{align}
\noindent From \eqref{eq:Lyap_8} and \eqref{eq:Lyap_9}, it follows that
\beq\label{eq:Lyap_10}
\gamma_1 \EE[\lVert x(k) \rVert^2] \leq (1-\gamma_4)^k \gamma_2 \EE[\lVert \xi_0\rVert_{\infty}^2].
\eeq
\noindent From \eqref{eq:Lyap_10}, the following inequality holds
\begin{align}\label{eq:Lyap_11}
    \EE[\lVert x(k) \rVert^2] &\leq  (1-\gamma_4)^k \dfrac{\gamma_2}{\gamma_1} \EE[\lVert \xi_0\rVert_{\infty}^2]\nonumber\\
                              & \leq  (1-\gamma_4)^k \dfrac{\gamma_2}{\gamma_1} \lVert \xi_0\rVert_{\infty}^2.
\end{align}
\noindent By defining $M\triangleq\dfrac{\gamma_2}{\gamma_1} \geq 1$ and $\zeta \triangleq (1-\gamma_4)$, with $0<\zeta<1$, from \eqref{eq:Lyap_11},
\beq\label{eq:Lyap_12}
\EE[\lVert x(k) \rVert^2] \leq M \zeta^k  (\lVert \xi_0\rVert_{\infty})^2.
\eeq
Thus, the system $\Sigma$ is {\it EMSS}.
\end{proof}
\subsection*{B. Proof of Theorem \ref{thm:msscomplete}}
\begin{proof}
\noindent From $(i)$-$(ii)$, by Lemma \ref{lem:1}, it follows 
 that there exist a function $V: \mathcal{C} \times D \to \RR^{+}$, $\gamma_{i}\in\RR^{+}$, $i=1,2,3$
such that, for all $\phi \in \mathcal{C},$ for all $i \in \mathcal{S}$, the following inequalities hold:
\begin{itemize}
	\item[$c_1)$] $\gamma_{1}\lVert \phi(0) \rVert ^2 \leq V(\phi, H^{-1}(i)) \leq \gamma_{2}\lVert \phi \rVert_{\infty} ^2$,
    \item[$c_2)$] $\mathcal{L} V (\phi, H^{-1}(i)) 	\leq - \gamma_3 \lVert \phi \rVert_{\infty}^{2},$
\end{itemize}
\noindent with $\mathcal{L}V$ defined in \eqref{eq:mathcalLVdef}.

\noindent From Lemma \ref{lem:mssnew}, we obtain that the system $\Sigma$ is {\it EMSS}.
\end{proof}
\end{document}